\documentclass[12pt,twoside]{article}

 \usepackage{float}
\usepackage{graphicx}
\usepackage{epstopdf}
\usepackage{graphicx}
\usepackage{epic}
\usepackage{multirow}
\usepackage{pst-poly}  
\usepackage{pst-plot}  
\usepackage{pst-poly}  
\usepackage{tikz}
\usepackage{xcolor}
\usetikzlibrary{arrows,shapes,chains}

\renewcommand{\paragraph}{\roman{paragraph}}
\usepackage[a4paper]{geometry}
\makeatletter
\renewcommand\title[1]{\gdef\@title{\reset@font\Large\bfseries #1}}
\renewcommand\section{\@startsection {section}{1}{\z@}%
                                   {-3.5ex \@plus -1ex \@minus -.2ex}%
                                   {2.3ex \@plus.2ex}%
                                   {\normalfont\large\bfseries}}
\renewcommand\subsection{\@startsection{subsection}{2}{\z@}%
                                     {-3ex\@plus -1ex \@minus -.2ex}%
                                     {1.5ex \@plus .2ex}%
                                     {\normalfont\normalsize\bfseries}}
\renewcommand\subsubsection{\@startsection{subsubsection}{3}{\z@}%
                                     {-2.5ex\@plus -1ex \@minus -.2ex}%
                                     {1.5ex \@plus .2ex}%
                                     {\normalfont\normalsize\bfseries}}

\def\@runningauthor{}\newcommand{\runningauthor}[1]{\def\runningauthor{#1}}
\def\@runningtitle{}\newcommand{\runningtitle}[1]{\def\runningtitle{#1}}

\renewcommand{\ps@plain}{%
\renewcommand{\@evenhead}{\footnotesize\scshape \hfill\runningauthor\hfill}
\renewcommand{\@oddhead}{\footnotesize\scshape \hfill\runningtitle\hfill}}

\newcommand{\F}{\mathbb{F}}
\newcommand{\x}{\mathbf{x}}

\newcommand{\T}{{\rm T}}

\newcommand {\D}{\mathfrak{D}}

\newcommand {\C}{{\mathcal{C}}}
\newcommand {\ccc}{{\mathbf{c}}}
\newcommand {\dd}{\mathbf{d}}
\newcommand {\0}{\mathbf{0}}
\newcommand {\aaa}{\alpha}
\newcommand {\bbb}{\beta}
\g@addto@macro\bfseries{\boldmath}

\makeatother

\usepackage{amsthm,amsmath,amssymb}
\usepackage{cite}
\usepackage{graphicx}

\usepackage[colorlinks=true,citecolor=black,linkcolor=black,urlcolor=blue]{hyperref}

\theoremstyle{plain}
\newtheorem{theorem}{Theorem}
\newtheorem{lemma}[theorem]{Lemma}
\newtheorem{corollary}[theorem]{Corollary}

\theoremstyle{definition}
\newtheorem{definition}[theorem]{Definition}
\newtheorem{example}[theorem]{Example}
\newtheorem{conjecture}[theorem]{Conjecture}
\newtheorem{open}[theorem]{Open Problem}

\newtheorem{prop}[theorem]{Proposition}
\theoremstyle{remark}
\newtheorem{remark}[theorem]{Remark}

\title{The $b$-symbol weight distribution of irreducible cyclic codes and related consequences
}

\runningtitle{The $b$-symbol weight distribution of irreducible cyclic codes and related consequences}

\author{Hongwei Zhu \thanks{ School of Mathematical Sciences, Anhui University, Hefei, China. E-mail: zhwgood66@163.com}
\and Minjia Shi\thanks{School of Mathematical Sciences, Anhui University, Hefei, China. E-mail: smjwcl.good@163.com}
}

\runningauthor{}

\date{}

\begin{document}

\maketitle

\thispagestyle{empty}

\begin{abstract}
The $b$-symbol read channel is motivated by the limitations of the reading process in high
density data storage systems.
The corresponding new metric is a generalization of the Hamming metric known as the $b$-symbol weight metric and has become an important object in coding theory.
In this paper, the general $b$-symbol weight enumerator formula for irreducible cyclic codes is presented by using the Gaussian period and a new invariant $\#U(b,j,N_1)$. The related $b$-symbol weight hierarchies $\{d_1(\C),d_2(\C),\ldots,d_K(\C)\}$ ($K=\dim(\C)$) are given for some cases. The shortened codes which are optimal from some classes of irreducible cyclic codes are given, where the shorten set $\mathcal{T}$ is the complementary set of $b$-symbol support of some codeword with the minimal $b$-symbol weight.
\end{abstract}

\section{Introduction}
The theory
of error-control codes aims to recover the original information units when some bound is given on their corruption. These
corruption bounds can be defined at the code-block level, such as a
certain number of errors in Hamming metric codes, or at the
individual-symbol level, such as symbol-transition restrictions in
asymmetric or unidirectional error-correcting codes. The alphabet on which the information unit is defined may change
throughout the coding problem, like in soft-decoding, but
it is still typically the same unit that is tracked and analyzed.
In 2011, Cassuto and Blaum \cite{CB1,CB} proposed a new coding framework for channels whose outputs are overlapping pairs of symbols. Such channels are motivated by storage applications in which the spatial resolution of the reader may be insufficient to isolate adjacent symbols. Codes are still defined by an alphabet, as usual. The goal is to protect against a certain number of pairwise errors, not a certain number of symbol errors. A pair-error is defined as a pair-read in which one or more of the symbols is read in error. Due to physical limitations, individual symbols cannot be read off the channel. Therefore, each channel read contains contributions from two adjacent symbols. The constructions of symbol-pair codes are studied in a series of papers \cite{CL,C+,C+1,CLL,Eli,DGZ,KZL,LG,ML1,ML,SZW,Yaa}.
Later, Yaakobi {\it et al.} \cite{Yaa1} generalized the symbol-pair read channel to the $b$-symbol read channel. The contributions to the $b$-symbol codes can be found in \cite{DZG,SOS,YLF,ZHW,ZHW1,ZHW2} and the references therein.

The definition of $b$-symbol metric will be introduced in detail in Section II of this paper. It is not hard to see that the $b$-symbol metric is a natural generalization of Hamming metric. Another generalization of Hamming metric is the $b$-th generalized Hamming metric, which has appeared as early as in 1970s \cite{Hell2,K1} and has become an important research topic in coding theory after the famous paper \cite{wei} in 1991, where Wei gave a series of wonderful consequences on the $b$-th generalized Hamming metric and indicates that it completely characterizes the performance of a linear code when it is used on the wire-tap channel of type II. For more details on the $b$-th generalized Hamming metric, we refer the readers to \cite{wei}.

 Is there any connection between the two types of generalizations? Shi {\it et al.} \cite{BUG} considered this question and showed many interesting results, especially if $\C$ is constayclic.
 Let $\C$ denote a linear code with dimension $k$. We use $\dd_b(\C)$ to denote the minimum $b$-th generalized Hamming distance of $\C$.
 When $b=1$, $\dd_1(\C)$ is the minimal Hamming distance of $\C$.
 The set $$\{\dd_b(\C)|1\leq b\leq k\}$$ is called the weight hierarchy of $\C$. To distinguish it from the later definition, let us call it the generalized weight hierarchy in the sequel.

 For a code $C$, $d_b(C)$ denotes the minimum $b$-symbol distance of $C$.
When $b=1$, $d_1(C)$ is also the minimum Hamming distance of $C$. The $b$-symbol metric is also called symbol-pair metric if $b=2$. The set
      $$\{d_b(C)|1\leq b\leq n\}$$
      is called the $b$-symbol weight hierarchy of $C$. Note that $C$ could be an unrestricted code under the $b$-symbol metric.
 If $C$ is a cyclic code (or a constacyclic code), then the $b$-symbol weight hierarchy of $C$ has the following property:
 \begin{equation}\label{equa1}
   d_1(C)<d_2(C)<\cdots<d_{k-1}(C)<d_k(C)=d_{k+1}(C)=\cdots=d_n(C)=n.
 \end{equation}
The generalized Hamming weight hierarchy of $C$ has a similar property to (\ref{equa1}), and $C$ could be a linear code but not cyclic.
\begin{theorem}\label{kgeqb}\cite{BUG}
If $\C$ is a cyclic code with length $n$ and dimension $K$, then $d_b(\C)=n$ for $K\leq b\leq n.$ Moreover, if $\C$ is a cyclic code, then $d_b(\C)=\dd_b(\C)$ if $b=1$ or $b=\dim(\C)$.
\end{theorem}
It is worth mentioning that there are other interesting connections between the two metrics. Liu and Pan \cite{LP1,LP2} considered the superposition of two metrics. They call this superposition a generalized $b$-weight (we prefer to call it a generalized $b$-symbol weight). For more details, we refer the readers to \cite{LP1,LP2}.

The following result shows that $\dd_b(\C)$ is a lower bound of $d_b(\C)$.
\begin{theorem}\label{dddd}\cite{BUG}
If $\C$ is a cyclic code, then $d_b(\C)\geq \dd_b(\C).$
\end{theorem}

It is very meaningful to determine the $b$-symbol weight hierarchy of cyclic codes, since it provides a nice upper bound on their generalized weight hierarchy. Moreover, there is another application for determining the $b$-symbol weight hierarchy of cyclic codes. If the $b$-symbol weight hierarchy of cyclic codes is known, we can obtain a shortened code by shortening some coordinates associated with the codeword with the minimum $b$-symbol weight. We will discuss this in detail in Section V of this paper.

Besides the weight hierarchy of cyclic codes, the Hamming weight structure of cyclic codes is also a hot topic in coding theory.
The Hamming weight structure of irreducible cyclic codes has been a research topic since the first works of McEliece and others \cite{LN,McE,D} due to their connection to Gaussian sums and $L$-functions, and its intrinsic complexity. As we all know, it is very difficult to determine the Hamming weight distribution of irreducible cyclic codes. Predictably, determining the $b$-symbol weight distribution of an irreducible cyclic code is even more difficult. To the best of our knowledge, there are only a few papers dealing with the $b$-symbol weight distribution of some cyclic codes:
\begin{itemize}
  \item Sun {\it et al.} \cite{SZW} considered the symbol-pair distance distribution of a class of repeated-root cyclic codes;
  \item Ma and Luo \cite{ML} considered the symbol-pair weight distribution of MDS codes and Simplex codes;
  \item Shi {\it et al.} \cite{SOS} gave some bounds on the $b$-symbol minimum distance of cyclic codes by a geometric approach;
  \item Zhu {\it et al.} considered the complete $b$-symbol weight distribution of a class of irreducible cyclic codes \cite{ZHW} and the $b$-symbol weight hierarchy of a class of reducible cyclic codes called Kasami codes \cite{ZHW1}.
\end{itemize}
This paper is a further study of the paper \cite{ZHW}. We give a formula for computing the $b$-symbol weight of a codeword of an irreducible cyclic code by using the Gaussian period and a new invariant $\#U(b,j,N_1)$. The definitions of the Gaussian period and $\#U(b,j,N_1)$ will be defined in Section II and Section III, respectively. The formula is a generalization of the formula for computing the Hamming weight of a codeword of an irreducible cyclic code given in \cite{DY2013}. We consider the $b$-symbol weight hierarchy of some irreducible cyclic codes. In particular, the two types of weight hierarchies of the same irreducible cyclic code are equal under some restrictions. Some optimal shortened codes are obtained by shortening some special coordinates, where these special coordinates are related to the codeword with the minimum $b$-symbol distance.

The paper is organized as follows. In Section II, we introduce various notations, definitions, and basic facts. Then, in Section III, we present a general formula for the $b$-symbol weight distribution of irreducible cyclic codes and some specific cases. In Section IV, we compute the $b$-symbol weight hierarchy of some classes of irreducible cyclic codes and compare these results with the known results on the generalized hierarchy of irreducible cyclic codes. In Section V, we present an application of the $b$-symbol weight hierarchy of cyclic codes in the shortening technique and construct some new shortened codes with nice parameters. Section VI concludes this paper.
\section{Preliminaries}
Throughout this paper we assume and fix the following:
\begin{itemize}
  \item Let $q=p^s$, $Q=q^m$, where $p$ is a prime number, $s,m$ are positive integers.
  \item Let $n$ denote the length of the code, where $n|Q-1$ and $\gcd(n,q)=1$. Let $k_0$ be the multiplicative order of $q$ modulo $n$ and $k_0| m$.
  \item Let $N=\frac{Q-1}{n}$, $\alpha$ be a primitive element of $\F_Q$ and $\theta=\alpha^{N}$.
  \item Let $T_{Q/q}$ denote the trace function from $\F_Q$ to $\F_q$.
  \item Let $supp(\x)$ denote the support of the vector $\x$.
\end{itemize}
\subsection{The $b$-symbol metric}
Let $b$ be a positive integer with $1\leq b\leq n.$ For any $\x\in\F_q^n$, the Hamming weight $w_H(\x)$ is defined as the number of nonzero coordinates in $\x$. Let $\pi_b(\x)$ denote the vector
$$\pi_b(\x)=((x_0,\ldots,x_{b-1}),(x_{1},\ldots,x_{b}),
\cdots,(x_{n-1},\ldots,x_{b+n-2}))\in \left(\F_q^b\right)^n,$$
where the indices are taken modulo $n$.
 The $b$-symbol weight of $\x$ is defined as
$$w_b(\x)=w_H(\pi_b(\x)).$$
\begin{example}
Let $\x=(0,0,a,0,0,0,b,0,0,0,0,c,0,a)\in\F_q^{14},$ where $a,b,c \in\F_q^*.$ Then the $b$-symbol weight of $\x$ are the following.
\begin{itemize}
  \item [{\rm(i)}] $w_1(\x)=w_H(\x)=4;$
  \item [{\rm(ii)}] $w_2(\x)=w_H(\pi_2(\x))=w_H\big((0,0),(0,a),(a,0),(0,0),(0,0),(0,b),
      (b,0),(0,0),(0,0),\\(0,0),(0,c),(c,0),(0,a),(a,0)\big)=8;$
  \item [{\rm(iii)}]
  $w_3(\x)=w_H(\pi_3(\x))=w_H\big((0,0,a),(0,a,0),(a,0,0),(0,0,0),(0,0,b),
      (0,b,0),\\(b,0,0),(0,0,0),(0,0,0),(0,0,c),(0,c,0),
      (c,0,a),(0,a,0),(a,0,0)\big)=11$;
  \item [{\rm(iv)}]
  $w_4(\x)=w_H(\pi_4(\x))=w_H\big((0,0,a,0),(0,a,0,0),(a,0,0,0),(0,0,0,b),(0,0,b,0),
      (0,b,\\0,0),(b,0,0,0),(0,0,0,0),(0,0,0,c),(0,0,c,0),(0,c,0,a),
      (c,0,a,0),(0,a,0,0),(a,0,\\0,a)\big)=13$;
  \item [{\rm(iv)}] $w_b(\x)=14$ if $b\geq 5$.   
\end{itemize}
\end{example}
For any $\x,\mathbf{y}\in\F_q^n$, the $b$-symbol distance between $\x$ and $\mathbf{y}$ is defined as
$$d_b(\x,\mathbf{y})=w_b(\x-\mathbf{y}).$$
When $b=1$, $w_1(\x)=w_H(\x)$ and $d_1(\x,\mathbf{y})=d_H(\x,\mathbf{y}).$ For convenience, we adopt $w_1(\x)$ and $d_1(\x,\mathbf{y})$ to represent the Hamming weight of $\x$ and the Hamming distance between $\x$ and $\mathbf{y}$, respectively. Let $E$ be a subset of $\F_q^n$. The minimum $b$-symbol distance of $d_b(E)$ is defined as
$$d_b(E)=\min\{d_b(\x,\mathbf{y})|\x, \mathbf{y}\in E {\hbox{~and~} } \x\neq \mathbf{y}\}.$$

A linear $[n,K,d_b(\C)]$ code $\C$ over $\F_q$ is a $K$-dimensional subspace of $\F_q^n$ with minimum $b$-symbol distance $d_b(\C)$. Let $A_i^{b}$ denote the number of codewords with $b$-symbol weight $i$ in a code of length $n$. The $b$-symbol weight enumerator of $\C$ is defined by
$$1+A_{1}^{b}T+A_{1}^{b}T^2+\cdots+A_{n}^{b}T^n.$$
In fact, the $b$-symbol weight of a nonzero vector will never less than $b$ by the definition of $b$-symbol metric. Therefore, the $b$-symbol weight enumerator of $\C$ is better to write as
$$1+A_b^{b}T^b+\cdots+A_{n}^bT^n.$$
\subsection{Cyclic codes}
Let $\tau(x_0,x_1,\ldots,x_{n-1})$ denote the vector $(x_{n-1},x_0, \ldots,x_{n-2})$ obtained from $(x_0,x_1,\ldots,x_{n-1})$ by the cyclic shift of the coordinates $i\mapsto i+1 {~\rm mod~} n$.
A linear $[n,k]$ code $\C$ over $\F_q$ is called cyclic if $\ccc\in \C$ implies $\tau(\ccc)\in \C.$
 Let $\gcd(n,q)=1.$
The set
\begin{equation}\label{eq1}
  \C(Q,N)=\left\{\ccc(\beta)=(T_{Q/q}(\beta),T_{Q/q}(\beta\theta)),\ldots,
T_{Q/q}(\beta\theta^{n-1})|\beta\in\F_Q\right\}
\end{equation}
is called an irreducible cyclic code over $\F_q$ with parameters $[n,k_0]$.
It is worth mentioning that the celebrated Golay code is an irreducible cyclic code and was used on the Mariner Jupiter-Saturn Mission.
An irreducible cyclic code is said to be semi-primitive if $n=\frac{Q-1}{N}$ where $N>2$ divides $q^j+1$ for some $j\geq1.$

\subsection{Group character, Gaussian sum, Gaussian periods}
An additive character of $\F_q$ is a nonzero function $\chi$ from $\F_q$ to the set of complex numbers such that $\chi(x+y)=\chi(x)\chi(y)$ for any $(x,y)\in\F_q^2.$ For each $b\in\F_q$, the function
$$\chi_b(c)=e^{2\pi\sqrt{-1}T_{q/p}(bc)/p}, ~~{\hbox{for all $c\in\F_q$}}$$
defines an additive character of $\F_q.$ $\chi_1$ is called the canonical additive character of $\F_q$.

A multiplicative character of $\F_q$ is a nonzero $\psi$ from $\F_q^*$ to the set of complex numbers such that $\psi(xy)=\psi(x)\psi(y)$ for all pairs $(x,y)\in\F_q^*\times\F_q^*$. Let $g$ be a fixed primitive element of $\F_q$. For each $j\in\{1,2,\ldots,q-1\}$, the function $\psi_j$ with
$$\psi_j(g^k)=e^{2\pi \sqrt{-1}jk/(q-1)}, \hbox{~~for $k\in\{0,1,2,\ldots,q-1\}$}$$
defines a multiplicative character with order $\frac{q-1}{\gcd(q-1,j)}$ of $\F_q$.

Let $\psi$ be a multiplicative character with order $k$ where $k|(q-1)$ and $\chi$ an additive character of $\F_q$. Then the Gaussian sum $G(\psi,\chi)$ of order $k$ is defined by
$$G(\psi,\chi)=\sum_{c\in \F_q^*}\psi(c)\chi(c).$$
For convenience, let $G(\psi)$ denote $G(\psi,\chi_1)$ in the sequel.

Let $C_i^{(k,Q)}=\alpha^i\langle\alpha^k\rangle$ for $i\in\{0,1,2,\ldots,k-1\}$, where $\langle\alpha^k\rangle$ denotes the subgroup of $\F_Q^*$ generated by $\alpha^k$. The cosets $C_i^{(k,Q)}$ are called the cyclotomic classes of order $k$ in $\F_Q$.

The Guassian periods are defined by
$$\eta_i^{(k,Q)}=\sum_{x\in C_{i}^{(k,Q)}}\chi_1(x), ~~~i\in\{0,1,\ldots,k-1\}.$$
By the discrete Fourier transform, Gaussian periods and Gaussian sum have the following relationship,
$$\eta_i^{(k,Q)}=\frac{-1+\sum_{j=1}^{k-1}\xi_k^{-ij}G(\psi^j)}{k},$$
where $\xi_k=e^{2\pi\sqrt{-1}}/k$ and $\psi$ is a primitive multiplicative character of order $k$ over $\F_Q^*.$

\begin{lemma}\label{lemeta}
Let symbols be the same as before. Then we have
\begin{itemize}
  \item [{\rm (1)}]\cite{Stor} $\sum_{i=0}^{k-1}\eta_i^{(k,Q)}=-1.$
  \item [{\rm (2)}]\cite{Stor} $\sum_{i=0}^{k-1}\left(\eta_i^{(k,Q)}\right)^2=Q\theta_j-\frac{Q-1}{k}$ for all $j=\{0,1,\ldots,k-1\}$, where
\begin{equation*}
  \theta_j=\left\{
             \begin{array}{ll}
               1 & \hbox{if $\frac{Q-1}{k}$ is even and $j=0$} \\
               1, & \hbox{if $\frac{Q-1}{k}$ is odd and $j=\frac{k}{2}$}\\
               0, & \hbox{otherwise,}
             \end{array}
           \right.
\end{equation*}
and equivalently $\theta_j=1$ if and only if $-1\in C_j^{(k,Q)}.$
  \item [{\rm (3)}]$\eta_0^{(k,Q)}, \eta_1^{(k,Q)}, \ldots, \eta_{k-1}^{(k,Q)}$ can not be all the same if and only if $k\geq 2.$
  \item [{\rm (4)}]Let $f(x)=\sum_{i=0}^{k-1}\eta^{(k,Q)}_ix^i$ and $\omega$ denote a primitive $k$-th root of unity.
   The circulant matrix
  \begin{equation}\label{circulant}
    A=\left(
       \begin{array}{cccc}
         \eta_0^{(k,Q)} & \eta_1^{(k,Q)} & \cdots & \eta_{k-1}^{(k,Q)} \\
         \eta_{k-1}^{(k,Q)} & \eta_0^{(k,Q)} & \cdots & \eta_{k-2}^{(k,Q)} \\
         \vdots & \vdots & \ddots & \vdots \\
         \eta_1^{(k,Q)} & \eta_2^{(k,Q)} & \cdots & \eta_{0}^{(k,Q)}
       \end{array}
     \right)_{k\times k}
  \end{equation}
  is invertible if and only if $f(\omega^i)\neq0$ for all $i\in \{0,1,\ldots,k-1\}.$
\end{itemize}
\end{lemma}
\begin{proof}
The first two properties of Gaussian periods are from \cite{Stor}. Assume that
\begin{equation}\label{etaj}
  \eta_0^{(k,Q)}=\eta_1^{(k,Q)}=\cdots=\eta_{k-1}^{(k,Q)}=\lambda.
\end{equation} From assertions (1) and (2) of this Lemma, we have
\begin{equation*}
  ~\left\{
     \begin{array}{ll}
       k\lambda=-1, \\
       k\lambda^2=Q-\frac{Q-1}{k}.
     \end{array}
   \right.
\end{equation*}
Then the Eq. (\ref{etaj}) holds if and only if $k=1$.

Let $P$ denote the permutation matrix
\begin{equation*}
  P=\left(
    \begin{array}{ccccc}
      0 & 1 & 0 & \cdots & 0 \\
      0 & 0 & 1 & \cdots & 0 \\
      \vdots & \vdots & \vdots & \vdots & \vdots \\
      0 & 0 & 0 & \cdots & 1 \\
      1 & 0 & 0 & \cdots & 0 \\
    \end{array}
  \right).
\end{equation*}
Then $P^k=I_k$ and $A=\eta_0^{(k,Q)} I_k+\eta_1^{(k,Q)}P+\cdots+\eta_{k-1}^{(k,Q)}P^{k-1},$ where $I_k$ denotes the unit matrix. Therefore, the circulant matrix $A$ is invertible if and only if $f(x)$ is coprime to $x^k-1.$
 This completes the proof.
\end{proof}
Since the values of the Gaussian sums in general are very hard to compute, the values of $\eta_i^{(k,Q)}$ are also hard to compute. Some known results on $\eta_i^{(k,Q)}$ are the following \cite{Ber,Mye}.
\begin{lemma}\label{lem1}
The known results on the Gaussian periods are the following.
\begin{enumerate}
  \item If $k=2$, then
\begin{equation*}
  \eta_0^{(2,Q)}=\left\{
                   \begin{array}{ll}
                     \frac{-1+(-1)^{sm-1}Q^{\frac{1}{2}}}{2}, & \hbox{if $p\equiv1~({\rm{mod}~4})$} \\
                     \frac{-1+(-1)^{sm-1}(\sqrt{-1})^{sm}
                     Q^{\frac{1}{2}}}{2}, & \hbox{if $p\equiv3~({\rm{mod}~4})$}
                   \end{array}
                 \right.
\end{equation*}
and $\eta_1^{(2,Q)}=-1-\eta_0^{(2,Q)}.$
  \item  If $k=3$ and $p\equiv2({\rm{mod}~3})$, then

\begin{equation}\label{3Q2}
  \eta_0^{(3,Q)}=\frac{-1-(\sqrt{-1})^{sm}2Q^{\frac{1}{2}}}{3},~
  \eta_1^{(3,Q)}=\eta_2^{(3,Q)}=\frac{-1+(\sqrt{-1})^{sm}Q^{\frac{1}{2}}}{3}.
\end{equation}

  \item If $k=4$ and $p\equiv3({\rm{mod}~4})$, then
          \begin{equation}\label{4Q2}
  \eta_0^{(4,Q)}=\frac{-1-(\sqrt{-1})^{sm}3Q^{\frac{1}{2}}}{4},
  \eta_1^{(4,Q)}=\eta_2^{(4,Q)}=\eta_3^{(4,Q)}=
  \frac{-1+(\sqrt{-1})^{sm}Q^{\frac{1}{2}}}{4}.
          \end{equation}
  \item $($Semi-primitive case$)$ If $k>2$ and there exists a positive integer $j$ such that $p^j\equiv-1({{\rm{mod}}~k})$, and the $j$ is the least such. Let $Q=p^{2j\gamma}$ for some integer $\gamma$.
      \begin{itemize}
        \item When $\gamma, p$ and $\frac{p^j+1}{k}$ are all odd, then
\begin{equation}\label{eta7}
  \eta_{\frac{k}{2}}^{(k,Q)}=\frac{(k-1)Q^{\frac{1}{2}}-1}{k},~
  \eta_{i}^{(k,Q)}=-\frac{Q^{\frac{1}{2}}+1}{k} \hbox{~~for $i\neq \frac{k}{2}.$}
\end{equation}
        \item In all other cases,
        \begin{equation}\label{eta8}
  \eta_{0}^{(k,Q)}=\frac{(-1)^{\gamma+1}(k-1)Q^{\frac{1}{2}}-1}{k},
  \eta_{i}^{(k,Q)}=\frac{(-1)^{\gamma}Q^{\frac{1}{2}}-1}{k} \hbox{~~for $i\neq 0.$}
\end{equation}
      \end{itemize}
\end{enumerate}
\end{lemma}
\section{The complete $b$-symbol weight enumerators}
Ding and Yang \cite{DY2013} proved that the determination of the Hamming weight distribution of an irreducible cyclic codes is equivalent to that of the Gaussian periods of order $\gcd\left(\frac{Q-1}{q-1},N\right)$. McEliece \cite{McE} gave another proof by Gaussian sums.
\begin{lemma}{\rm \cite{DY2013,McE}}\label{Dingyang2013}
Let $\ccc(\beta)$ be a codeword of the irreducible cyclic code $\C(Q,N)$ as in (\ref{eq1}).
If $0\neq\beta\in C_i^{\left(\gcd\left(\frac{Q-1}{q-1},N\right),Q\right)}$, then the Hamming weight of $\ccc(\beta)$ is
$$w_1(\ccc(\beta))=\frac{(q-1)(Q-1)}{qN}-\frac{(q-1)\gcd\left(
\frac{Q-1}{q-1},N\right)\eta_i^{\left(\gcd\left(
\frac{Q-1}{q-1},N\right),Q\right)}}{qN}.$$
\end{lemma}
Very recently, Shi {\it et al.} \cite{BUG} studied the relationship between the $b$-th generalized Hamming weight metric and $b$-symbol weight metric. A very interesting expression on the $b$-symbol weight of a vector $\ccc$ is given in that paper, and we present it in the following.
It is very important for giving the expression of the $b$-symbol weight of a codeword in $\C(Q,N)$.
\begin{lemma}{\rm\cite{BUG}}\label{lem3}
Let $\mathbf{c}\in \F_q^n$ and denote by $V_b(\mathbf{c})$ the codewords generated by all linear combinations of $\ccc$ and its first $b-1$ cyclic shifts. Then
$$w_b(\mathbf{c})=\frac{1}{q^{b-1}(q-1)}\sum_{\mathbf{c}^{\prime}\in V_b(\mathbf{c})}w_1(\mathbf{c}^{\prime}).$$
\end{lemma}

Combining the two lemmas above, we obtain the following result, which is the key observation of this paper.
For convenience, let $N_1=\gcd\left(\frac{Q-1}{q-1},N\right)$  in the sequel.

\begin{theorem}\label{thm4}
Let $\ccc(\beta)$ be a codeword of the irreducible cyclic code $\C(Q,N)$ as in (\ref{eq1}).
Let $\alpha$ be a primitive element of $\F_Q$, $\theta=\alpha^N$ and $(u_1,\ldots,u_b)\in\F_q^b\setminus\{\mathbf{0}\}$. Assume that $1\leq b\leq k_0-1$ and the nonzero elements $\sum_{i=1}^bu_i\theta^{i-1}$ belongs to the cyclotomic class $C^{(N_1,Q)}_{k_{(u_1,\ldots,u_b)}}$.
If $0\neq \beta\in C_i^{(N_1,Q)}$, then the $b$-symbol weight of $\ccc(\beta)$ is
\begin{equation*}
  w_b(\ccc(\beta))=\frac{(q^b-1)(Q-1)}{q^bN}-
\frac{N_1}{q^bN}\sum_{(u_1,\ldots,u_b)\in \F_q^b\setminus\{\mathbf{0}\}}
\eta^{(N_1,Q)}_{i+k_{(u_1,\ldots,u_b)}},
\end{equation*}
where the indices are taken modulo $N_1$.
\end{theorem}
\begin{proof}
 For any nonzero element $\beta\in\F_Q$, according to the definition of $V_b(\ccc(\beta))$ in Lemma \ref{lem3}
and the property of the trace function, we have
\begin{eqnarray*}
  V_b(\ccc(\beta)) &=& \left\{\left.\sum_{j=1}^bu_j\tau(\ccc(\beta))\right|
  (u_1,\ldots,u_b)\in\F_q^b\right\} \\
  ~ &=& \left\{\left.\sum_{j=1}^bu_j\ccc(\beta\theta^{j-1})\right|
  (u_1,\ldots,u_b)\in\F_q^b\right\} \\
  ~ &=& \left\{\left.\ccc(\beta\sum_{j=1}^bu_j\theta^{j-1})\right|
  (u_1,\ldots,u_b)\in\F_q^b\setminus\{\mathbf{0}\}\right\}\cup
  \{\mathbf{0}\}.
\end{eqnarray*}
Since $\beta\in C_i^{(N_1,Q)}$ and $\sum_{j=1}^bu_j\theta^{j-1}\in C^{(N_1,Q)}_{k_{(u_1,\ldots,u_b)}}$,  we have $$\beta\sum_{j=1}^bu_j\theta^{j-1}\in C^{(N_1,Q)}_{i+k_{(u_1,\ldots,u_b)}},$$
where the indices are taken modulo $N_1$.
 Combining Lemma \ref{Dingyang2013} and Lemma \ref{lem3}, the $b$-symbol weight of $\ccc(\beta)$ equals
\begin{eqnarray*}
  w_b(\ccc(\beta)) &=& \frac{1}{q^{b-1}(q-1)}\left(\sum_{\ccc^{\prime}\in V_b(\ccc(\beta))\setminus\{\mathbf{0}\}}w_1(\ccc^{\prime})+w_1(\mathbf{0})\right) \\
  ~ &=& \frac{(q^b-1)(Q-1)}{q^bN}- \sum_{(u_1,\ldots,u_b)\in \F_q^b\setminus\{\mathbf{0}\}}
  \frac{N_1\eta^{(N_1,Q)}_{i+k_{(u_1,\ldots,u_b)}}}{q^bN},
\end{eqnarray*}
where the indices are taken modulo $N_1$.
This completes the proof.
\end{proof}
\begin{definition}\label{UB}
Define $U(b,i,N_1)$ be the set
\begin{equation*}\label{UBIN}
  U(b,i,N_1)=\left\{(u_1,\ldots,u_b)\left|
  \sum_{i=1}^{b}u_i\theta^{i-1}\in C_i^{(N_1,Q)}\right. {\hbox{~and~}} (u_1,\ldots,u_b)\in\F_q^b\setminus\{\mathbf{0}\}\right\},
\end{equation*}
where $\theta=\alpha^N.$
\end{definition}
\begin{lemma}\label{lem7}
We have the following properties on $U(b,i,N_1)$.

\begin{itemize}
  \item[ {\rm (1)}] $\F_q^b=\{\mathbf{0}\}\cup\bigcup_{i=0}^{N_1-1}U({b,i,N_1}).$
  \item[ {\rm (2)}] $\sum_{i=0}^{N_1-1}\#U(b,i,N_1)=q^b-1$, where $\#U(b,j,N_1)$ denotes the size of $U(b,j,N_1).$
  \item[ {\rm (3)}] $U(1,0,N_1)=\F_q^*$ and $U(1,j,N_1)=\emptyset$ for all $j\in\{1,2,\ldots,N_1-1\}$.
  \item[ {\rm (4)}]$\#U(k_0,i,N_1)=\frac{q^{k_0}-1}{N_1}$ for all $i\in\{0,1,\ldots,N_1-1\},$ where $k_0$ is the multiplicative order of $q$ modulo $n$.
      \item[ {\rm (5)}] $\#U(b,0,N_1)\geq b(q-1)$. Moreover, $\#U(b,i,N_1)\leq \frac{Q-1}{N_1}$  and
          $\#U\left(b,0,\frac{Q-1}{(q-1)b}\right)=b(q-1)$ if $b\leq k_0.$
  \item[ {\rm (6)}] $\F_q^*=U(1,0,N_1)\subset U(2,0,N_2)
  \subset\cdots\subset U(k_0,0,N_1)=C_0^{(N_1,q^{k_0})}$ and
  $\emptyset=U(1,j,N_1)\subset U(2,j,N_2)\subset\cdots\subset U(k_0,j,N_1)=C_j^{(N_1,q^{k_0})}$ for all $j\in \{1,2,\ldots,N_1-1\}.$
\end{itemize}
\end{lemma}
\begin{proof}
The first two statements are trivial.
When $b=1$, $u_1$ has to be a nonzero elements of $\F_q$. Since $\F_q^*=\langle\alpha^{\frac{Q-1}{q-1}}\rangle$ and $N_1=\gcd\left(\frac{Q-1}{q-1},N\right)$, we obtain $u_1\in C_0^{(N_1,Q)}$, $k_{(u_1)}=0$, $U(1,0,N_1)=\F_q^*$ and $U(1,i,N_1)=\emptyset$ for $i\neq0$.

When $b=k_0,$ the set $\left\{\left.\sum_{i=1}^{k_0}u_i\theta^{i-1}\right|(u_1,\ldots,u_{k_0})
\in\F_q^{k_0}\setminus\{\mathbf{0}\}\right\}=\F_{q^{k_0}}^*$ since $k_0$ is the multiplicative order of $q$ modulo $n$. Then $\#U(k_0,i,N_1)$ equals the size of the following set $$\left\{\alpha\left|\alpha\in C_i^{\left(N_1,{Q}\right)}\cap\F_{q^{k_0}}^*\right.\right\}=
\left\{\alpha\left|\alpha\in C_i^{\left(N_1,{q^{k_0}}\right)}\right.\right\}.$$
Therefore, $\#U(k_0,i,N_1)=\frac{q^{k_0}-1}{N_1}$ for all $i\in\{0,1,\ldots,N_1-1\}.$

For any $i\in \{1,2,\ldots,b\}$, we have $u_i\theta^{i-1}\in C_0^{(N_1,Q)}$ if $u_i\neq 0$. Then
$$(0,\ldots,0,u_i,0,\ldots,0)\in U(b,0,N_1).$$
Therefore, $\#U(b,0,N_1)\geq b(q-1).$ If $b\leq k_0$, then $\#U(b,i,N_1)\leq \left|C_i^{(N_1,Q)}\right|=\frac{Q-1}{N_1}.$ Moreover, $$b(q-1)\leq\#U\left(b,0,\frac{Q-1}{(q-1)b}\right)\leq \frac{Q-1}
{\frac{Q-1}{(q-1)b}}=b(q-1).$$
Therefore, $\#U\left(b,0,\frac{Q-1}{(q-1)b}\right)=b(q-1).$

Combining the parts (3), (4) and (5) of this lemma, we obtain the last desired result.
\end{proof}
\begin{example}
The numerical examples in  Table 1 are computed by Magma. In these examples, we let $b\left|\frac{Q-1}{q-1}\right.$, $1\leq b\leq k_0$ and $N_1=\frac{Q-1}{q-1}$. The value of $\#U(b,0,N_1)$ computed by Magma is consistent with the part (5) of Lemma \ref{lem7}.

\begin{table}[H]
  \centering
  \caption{Numeral examples of Lemma \ref{lem7}}\label{A}

\begin{tabular}{c|c|c|c|c|c}
  \hline
  $Q$ & $q$ & $b$ & $N$ & $N_1$ & $\#U(b,0,N_1)$ \\
  \hline
  $2^4$ & $2$ & $3$ & $5$ & $5$ & $3$ \\
  $2^6$ & $2$ & $3$ & $21$ & $21$ & $3$ \\
  $2^8$ & $2$ & $5$ & $51$ & $51$ & $5$ \\
  $2^{10}$ & $2$ & $3$ & $341$ & $341$ & $3$ \\
  $4^{6}$ & $4$ & $3$ & $455$ & $455$ & $9$ \\
  $4^{6}$ & $4$ & $5$ & $273$ & $273$ & $15$ \\
  $4^{8}$ & $4$ & $5$ & $4369$ & $4369$ & $15$ \\
  $3^{4}$ & $3$ & $2$ & $20$ & $20$ & $4$ \\
  $3^{6}$ & $3$ & $2$ & $182$ & $182$ & $4$ \\
  $3^{8}$ & $3$ & $2$ & $1640$ & $1640$ & $4$ \\
  $3^{8}$ & $3$ & $4$ & $820$ & $820$ & $8$ \\
  $3^{8}$ & $3$ & $5$ & $656$ & $656$ & $10$ \\
  \hline
\end{tabular}
\end{table}
\end{example}
Therefore, the determination of the $b$-symbol weight distribution of an irreducible cyclic code is equivalent to the values of $\#U(b,i,N_1)$ and $\eta_i^{(N_1,Q)}$. The following result is a generalization of Lemma \ref{Dingyang2013}.

\begin{corollary}
Let $1\leq b\leq k_0$ and let $\ccc(\beta)$ be a codeword of the irreducible cyclic code $\C(Q,N)$ as in (\ref{eq1}).
Let $\alpha$ be a primitive element of $\F_Q$, $\theta=\alpha^N$, $N_1=\gcd\left(\frac{Q-1}{q-1},N\right)$ and $(u_1,\ldots,u_b)\in\F_q^b\setminus\{\mathbf{0}\}$.
If $0\neq \beta\in C_i^{(N_1,Q)}$, then the $b$-symbol weight of $\ccc(\beta)$ is
\begin{equation*}
  w_b(\ccc(\beta))=\frac{(q^b-1)(Q-1)}{q^bN}-
\frac{N_1}{q^bN}\sum_{i=0}^{N_1-1}\#U(b,i,N_1)\eta_i^{(N_1,Q)}.
\end{equation*}
\end{corollary}
\begin{proof}
The desired result follows from the definition of $U(b,i,N_1)$ and Theorem \ref{thm4}.
\end{proof}
The two keys to determine the $b$-symbol weight distribution of irreducible cyclic codes are $\#U(b,i,N_1)$ and $\eta_i^{(N_1,Q)}$.
By Lemma \ref{lem7}, $\#U(1,i,N_1)$ are determined for $i\in \{0,1,\ldots,N_1-1\}$.
However, there are very few known results on $\#U(b,i,N_1)$ for $b\geq 2$. Under some restrictions on $b, i, N_1$ and $k_0$, Lemma \ref{lem7} gives an upper bound for $\#U(b,i,N_1)$, a lower bound for $\#U(b,0,N_1)$. When $b\leq k_0$ and $N_1=\frac{Q-1}{(q-1)b}$, the value of $\#U(b,i,N_1)$ is determined.
\begin{open}
Determine the value of $\#U(b,i,N_1)$ for all $i\in \{0,1,\ldots,N_1-1\}.$ Or give some strong bounds for $\#U(b,i,N_1)$.
\end{open}
The following theorem gives a necessary and sufficient condition for an irreducible cyclic code to be a constant $b$-symbol weight code under some assumption.
\begin{theorem}\label{thm8}
Let $1\leq b\leq m-1.$
Assume that the matrix $A$ defined as in (\ref{circulant}) is invertible. Then
the irreducible cyclic code $\C(Q,N)$ is a constant $b$-symbol code with length $\frac{Q-1}{N}$ and dimension $m$ if and only if $\#U(b,i,N_1)=\#U(b,j,N_1)$ for all $i\neq j.$ Moreover, the constant $b$-symbol weight is $\frac{(q^b-1)Q}{q^bN}.$
\end{theorem}
\begin{proof}
We prove the sufficient condition at first. Assume that $\#U(b,i,N_1)=\#U(b,j,N_1)$ for all $i\neq j.$ According to Theorem \ref{thm4} and the definition of $U(b,i,N_1)$, for any two nonzero elements $\beta_1,\beta_2\in\F_Q$,
\begin{itemize}
  \item if $\beta_1$ and $\beta_2$ belong to the same cyclotomic class, then $w_b(\ccc(\beta_1))=w_b(\ccc(\beta_2));$
  \item if $\beta_1\in C_i^{(N_1,Q)}$ and $\beta_2\in C_j^{(N_1,Q)}$ where $i\neq j$, then
\begin{eqnarray*}
  w_b(\ccc(\beta_1)) &=& \frac{(q^b-1)(Q-1)}{q^bN}-\frac{N_1}{q^bN}\sum_{k=0}^{N_1-1}
\#U(b,k,N_1)\eta_{i+k({\rm mod~N_1})}^{(N_1,Q)} \\
~&=&\frac{(q^b-1)(Q-1)}{q^bN}-\frac{N_1}{q^bN}\#U(b,0,N_1)\sum_{k=0}^{N_1-1}
\eta_{i+k({\rm mod~N_1})}^{(N_1,Q)}\\
~&=&\frac{(q^b-1)(Q-1)}{q^bN}-\frac{N_1\#U(b,0,N_1)}{q^bN}(-1)~~~
({\hbox{from Lemma \ref{lemeta}}})\\
~&=&\frac{(q^b-1)(Q-1)}{q^bN}+\frac{N_1\#U(b,0,N_1)}{q^bN}.
\end{eqnarray*}
Similarly, we can prove that
\begin{eqnarray*}
  w_b(\ccc(\beta_2)) &=& \frac{(q^b-1)(Q-1)}{q^bN}-\frac{N_1}{q^bN}\sum_{k=0}^{N_1-1}
\#U(b,k,N_1)\eta_{j+k({\rm mod~N_1})}^{(N_1,Q)} \\
~&=&\frac{(q^b-1)(Q-1)}{q^bN}+\frac{N_1\#U(b,0,N_1)}{q^bN}\\
~&=& w_b(\ccc(\beta_1)).
\end{eqnarray*}
\end{itemize}
We now prove the necessity of the condition. Assume that $\C(Q,N)$ is a constant $b$-symbol weight code and the constant $b$-symbol weight is $\zeta_1$. Then  for any $i\in\{0,1,\ldots,N_1-1\}$, $\sum_{k=0}^{N_1-1}
\#U(b,k,N_1)\eta_{i+k({\rm mod~N_1})}^{(N_1,Q)}$ is a constant which equals $\zeta_2.$
Assume that $\beta_i\in C_i^{(N_1,Q)}$ for $i \in \{0,1,\ldots,N_1-1\}$ and $A(i)$ denotes the matrix
\begin{equation*}
  A(i)=\left(
    \begin{array}{ccccccc}
      \eta_0^{(N_1,Q)} & \cdots & \eta_{i-2}^{(N_1,Q)} & \zeta_2 & \eta_{i}^{(N_1,Q)} & \cdots & \eta_{N_1-1}^{(N_1,Q)} \\
       \eta_{N_1-1}^{(N_1,Q)} & \cdots & \eta_{i-3}^{(N_1,Q)} & \zeta_2 & \eta_{i-1}^{(N_1,Q)} & \cdots & \eta_{N_1-2}^{(N_1,Q)} \\
      \vdots & \ddots & \vdots & \vdots & \vdots & \ddots & \vdots \\
       \eta_1^{(N_1,Q)} & \cdots & \eta_{i-1}^{(N_1,Q)} & \zeta_2 & \eta_{i+1}^{(N_1,Q)} & \cdots & \eta_{0}^{(N_1,Q)} \\
    \end{array}
  \right)_{N_1\times N_1}.
\end{equation*}
According to Lemma \ref{lemeta}, we have
\begin{eqnarray*}
  \sum_{i=0}^{N_1-1}w_b(\ccc(\beta_i)) &=& \frac{(q^b-1)(Q-1)N_1}{q^bN}-\frac{N_1}{q^bN}\sum_{i=0}^{N_1-1}
\sum_{k=0}^{N_1-1}\#U(b,k,N_1)\eta_{i+k({\rm mod }~N_1)}^{(N_1,Q)} \\
  ~ &=& \frac{(q^b-1)(Q-1)N_1}{q^bN}-\frac{N_1}{q^bN}(q^b-1)\cdot(-1) \\
  ~ &=& \frac{(q^b-1)QN_1}{q^bN}.
\end{eqnarray*}
Then $\zeta_1=\frac{(q^b-1)Q}{q^bN}$ and $\zeta_2=-q^b+1$. Therefore, the constant $b$-symbol weight is $\frac{(q^b-1)Q}{q^bN}$.
Solving the following system of equations, we obtain
\begin{equation}\label{eqs1}
  \left\{
    \begin{array}{ll}
      w_b(\ccc(\beta_0))=\zeta_1, \\
      w_b(\ccc(\beta_1))=\zeta_1, \\
      \vdots \\
      w_b(\ccc(\beta_{N_1-1}))=\zeta_1, \\
    \end{array}
  \right.
\Longrightarrow
A\cdot\left(
        \begin{array}{c}
          \#U(b,0,N_1) \\
          \#U(b,1,N_1) \\
          \vdots \\
          \#U(b,N_1-1,N_1) \\
        \end{array}
      \right)=\left(
                \begin{array}{c}
                  \zeta_2 \\
                  \zeta_2 \\
                  \vdots \\
                  \zeta_2 \\
                \end{array}
              \right),
\end{equation}
where the matrix $A$ is defined as in Lemma \ref{lemeta}.
  Combing the assumption that $A$ is invertible and the Cramer's rule, the solutions of (\ref{eqs1}) are
\begin{equation*}
  \#U(b,i,N_1)=\frac{\det(A(i))}{\det(A)} \hbox{~for all $i\in\{0,1,\ldots,N_1-1\}$}.
\end{equation*}
The desired result follows since $\det(A(0))=\cdots=\det(A(N_1-1)).$
\end{proof}

Theorem \ref{thm8} gives a complete characterization of constant $b$-symbol weight irreducible cyclic codes in the general case that $N$ is any divisor of $Q-1$. It easy to check that $\#U(b,i,N_1)=\#U(b,j,N_1)$ for all $i\neq j$ if $N_1=1$ or $b=m$. Shi {\it et al.} \cite{SOS} give the $b$-symbol weight enumerator of the irreducible cyclic codes $\C(Q,N)$ when $N_1=1$. By Theorem \ref{kgeqb}, for any cyclic code $\C$ with dimension $m$, $\C$ is a constant $b$-symbol weight code with $b$-symbol weight $n$ if $b\geq m$.

Recall that $f(x)=\sum_{i=0}^{k-1}\eta^{(k,Q)}_ix^i.$
From a number of numerical results computed by Magma, we find the assumption $f(w^i)\neq 0$ always holds for all $i \in\{0,1,\ldots,N_1-1\}$. It is reasonable to conjecture that $A$ is invertible by the part (4) of Lemma \ref{lemeta}. This conjecture is of interest for cyclotomy.
\begin{conjecture}
The circulant matrix $A$ defined as in (\ref{circulant}) is invertible.
\end{conjecture}
\subsection{The complete $b$-symbol weight enumerator when $N_1=2$}
Zhu {\it et al.} \cite{ZHW} considered the $b$-symbol weight enumerator of the irreducible cyclic codes $\C(Q,N)$ where $N_1=2$. For completeness, we list these results as follows.
\begin{theorem}\label{ZHWDCC}
\cite{ZHW}
Let $\mathcal{P}(b)$ be the subset of cardinality $\frac{q^b-1}{q-1}$ in $\F_{Q}^*$ defined as
\begin{equation*}
  \begin{array}{rcl}
\mathcal{P}(b)  =   \bigcup_{j=1}^{b-1} \left\{\theta^{(j-1)} + x_1 \theta^{j} + \cdots + x_{b-j} \theta^{(b-1)}| (x_1, \ldots, x_j) \in \F_q^j \right\}
  \cup \left\{ \theta^{(b-1)} \right\}
\end{array}
\end{equation*}
and let $$\mu(b)=\#\left\{ x  \in \mathcal{P}(b)| x \; \mbox{is a square in} \; \F_{Q}^*\right\}.$$
If $N_1=2$ and $1\leq b\leq m-1$, then the $b$-symbol weight distribution of $\C(Q,N)$ is
$$1+\frac{Q-1}{2}(T^{u_1}+T^{u_2}),$$
where
\begin{equation*}
  u_1=
  \left\{
    \begin{array}{ll}
      \frac{q^b-1}{N(q-1)q^{b-1}} \left( Q - \frac{Q+(q-1)Q^{\frac{1}{2}}}{q}\right)
+\frac{2 \mu(b) (q-1) Q^{\frac{1}{2}}}{Nq^b} & \mbox{if $p \equiv 1 \mod 4$}, \\
      \frac{q^b-1}{N(q-1)q^{b-1}} \left( Q - \frac{Q+(-1)^{\frac{sm}{2}}(q-1)Q^{\frac{1}{2}}}{q}\right)
+\frac{2 \mu(b) (-1)^{\frac{sm}{2}} (q-1) Q^{\frac{1}{2}}}{Nq^b} & \mbox{if $p \equiv 3 \mod 4$},
    \end{array}
  \right.
\end{equation*}
and
\begin{equation*}
  u_2=
  \left\{
    \begin{array}{ll}
      \frac{q^b-1}{N(q-1)q^{b-1}} \left( Q - \frac{Q-(q-1)Q^{\frac{1}{2}}}{q}\right)
-\frac{2 \mu(b) (q-1) Q^{\frac{1}{2}}}{Nq^b}& \mbox{if $p \equiv 1 \mod 4$}, \\
      \frac{q^b-1}{N(q-1)q^{b-1}} \left( Q - \frac{Q-(-1)^{\frac{sm}{2}}(q-1)Q^{\frac{1}{2}}}{q}\right)
-\frac{2 \mu(b) (-1)^{\frac{sm}{2}} (q-1) Q^{\frac{1}{2}}}{Nq^b}& \mbox{if $p \equiv 3 \mod 4$}.
    \end{array}
  \right.
\end{equation*}
\end{theorem}
The expressions of Theorem \ref{ZHWDCC} depends on the invariant $\mu(b)$. By the definitions of $\#U(b,0,2)$ and $\mu(b)$, we obtain $\#U(b,0,2)=(q-1)\mu(b).$ For the sake of consistency, we give the expression by using $\#U(b,0,2)$ rather than $\mu(b)$ in the following.
\begin{theorem}\label{thm1515}
 If $N_1=2$ and $1\leq b\leq m-1$, then the $b$-symbol weight distribution of $\C(Q,N)$ is
$$1+\frac{Q-1}{2}(T^{u_1}+T^{u_2}),$$
where
\begin{equation*}
  u_1=
  \left\{
    \begin{array}{ll}
      \frac{(q^b-1)(Q-Q^{\frac{1}{2}})}{Nq^{b}}
+\frac{2   Q^{\frac{1}{2}}}{Nq^b}\#U(b,0,2), & \hbox{if $p\equiv1 ({\rm mod~}4)$,} \\
      \frac{(q^b-1)(Q-(-1)^{\frac{sm}{2}}Q^{\frac{1}{2}})}{Nq^{b}}
+\frac{2 (-1)^{\frac{sm}{2}} Q^{\frac{1}{2}}}{Nq^b}\#U(b,0,2), & \hbox{if $p\equiv3 ({\rm mod~}4)$,}
    \end{array}
  \right.
\end{equation*}
and
\begin{equation*}
  u_2=
  \left\{
    \begin{array}{ll}
      \frac{(q^b-1)(Q+Q^{\frac{1}{2}})}{Nq^{b}}
-\frac{2  Q^{1/2}}{Nq^b}\#U(b,0,2), & \hbox{if $p\equiv1 ({\rm mod~}4)$,} \\
      \frac{(q^b-1)(Q+(-1)^{\frac{sm}{2}}Q^{\frac{1}{2}})}{Nq^{b}}
-\frac{2 (-1)^{\frac{sm}{2}} Q^{\frac{1}{2}}}{Nq^b}\#U(b,0,2), & \hbox{if $p\equiv3 ({\rm mod~}4)$.}
    \end{array}
  \right.
\end{equation*}
Moreover, \begin{equation*}
  w_b(\ccc(\beta))=\left\{
                     \begin{array}{ll}
                       0, & \hbox{if $\beta=0$;} \\
                       u_1, & \hbox{if $\beta\in C_0^{(2,Q)}$;} \\
                       u_2, & \hbox{if $\beta\in C_1^{(2,Q)}$.}
                     \end{array}
                   \right.
\end{equation*}
\end{theorem}

\subsection{The complete $b$-symbol weight enumerator when $N_1=3$}
\begin{theorem}\label{thm11}
Let $N_1=3$ and $1\leq b\leq m-1$.
When $p\equiv2({\rm mod~3})$, the $b$-symbol weight enumerator of $\C(Q,N)$ is
\begin{equation*}
  A_b(T)=1+\frac{Q-1}{3}(T^{u_1}+T^{u_2}+T^{u_3}),
\end{equation*}
where
\begin{equation*}
  u_i=\left\{
        \begin{array}{ll}
          \frac{(q^b-1)(Q-Q^{\frac{1}{2}})}{q^bN}
      +\frac{3Q^{\frac{1}{2}}}{q^bN}\#U(b,j,3), & \hbox{if $\frac{sm}{2}$ is even,} \\
      \frac{(q^b-1)(Q+Q^{\frac{1}{2}})}{q^bN}
      -\frac{3Q^{\frac{1}{2}}}{q^bN}\#U(b,j,3), & \hbox{if $\frac{sm}{2}$ is odd,}
        \end{array}
      \right.
\end{equation*}
and $i+j\equiv0~(~{\rm mod} ~3).$
Moreover,
\begin{equation*}
  w_b(\ccc(\beta))=\left\{
                     \begin{array}{ll}
                       0, & \hbox{if $\beta=0$;} \\
                       u_1, & \hbox{if $\beta\in C_0^{(3,Q)}$;} \\
                       u_2, & \hbox{if $\beta\in C_1^{(3,Q)}$;} \\
                       u_3, & \hbox{if $\beta\in C_2^{(3,Q)}$.}
                     \end{array}
                   \right.
\end{equation*}
\end{theorem}
\begin{proof}
Assume the $0\neq \beta \in C_i^{(3,Q)}$ and
$j\equiv-i~({\rm mod} ~3).$
When $p\equiv2~(~{\rm mod~3})$, from the known results on the Gaussian periods, $\eta_i^{(3,Q)}$ take only two values. Combining the values of $\eta_i^{(3,Q)}$ (see (\ref{3Q2}) of Lemma \ref{lem1}) and Theorem \ref{thm4}, we have
\begin{eqnarray*}
  w_b(\ccc(\beta)) &=& \frac{(q^b-1)(Q-1)}{q^bN}-\#U(b,j,3)\frac{3\eta_0^{(3,Q)}}{q^bN}\\
  ~&~&-(q^b-1-\#U(b,j,3))\frac{3\eta_{j_1}^{(3,Q)}}{q^bN}~~~~~~~(\hbox{where $j_1\neq 0$}) \\
  ~ &=&\left\{
    \begin{array}{ll}
      \frac{(q^b-1)(Q-Q^{\frac{1}{2}})}{q^bN}
      +\frac{3Q^{\frac{1}{2}}}{q^bN}\#U(b,j,3), & \hbox{if $\frac{sm}{2}$ is even;} \\
      \frac{(q^b-1)(Q+Q^{\frac{1}{2}})}{q^bN}
      -\frac{3Q^{\frac{1}{2}}}{q^bN}\#U(b,j,3), & \hbox{if $\frac{sm}{2}$ is odd.}
    \end{array}
  \right.
\end{eqnarray*}
This completes the proof.
\end{proof}
\subsection{The complete $b$-symbol weight enumerator when $N_1=4$}
\begin{theorem}\label{thm1717}
Let $N_1=4$ and $1\leq b\leq m-1$.
 When $p\equiv3({\rm mod~4})$, the $b$-symbol weight enumerator of $\C(Q,N)$ is
\begin{equation*}
  A_b(T)=1+\frac{Q-1}{4}(T^{u_1}+T^{u_2}+T^{u_3}+T^{u_4}),
\end{equation*}
where
\begin{equation*}
  u_i=\left\{
        \begin{array}{ll}
          \frac{(q^b-1)(Q-Q^{\frac{1}{2}})}{q^bN}
      +\frac{4Q^{\frac{1}{2}}}{q^bN}\#U(b,j,4) & \hbox{if $\frac{sm}{2}$ is even,} \\
      \frac{(q^b-1)(Q+Q^{\frac{1}{2}})}{q^bN}
      -\frac{4Q^{\frac{1}{2}}}{q^bN}\#U(b,j,4) & \hbox{if $\frac{sm}{2}$ is odd,}
        \end{array}
      \right.
\end{equation*}
and $i+j\equiv0~(~{\rm mod} ~4).$
  Moreover,
\begin{equation*}
  w_b(\ccc(\beta))=\left\{
                     \begin{array}{ll}
                       0, & \hbox{if $\beta=0$;} \\
                       u_1, & \hbox{if $\beta\in C_0^{(4,Q)}$;} \\
                       u_2, & \hbox{if $\beta\in C_1^{(4,Q)}$;} \\
                       u_3, & \hbox{if $\beta\in C_2^{(4,Q)}$;} \\
                       u_4, & \hbox{if $\beta\in C_3^{(4,Q)}$.}
                     \end{array}
                   \right.
\end{equation*}
\end{theorem}
\begin{proof}
Assume that $0\neq \beta\in C_i^{(4,Q)}$ and $i+j\equiv0({\rm mod~} 4).$ From (\ref{4Q2}) of Lemma \ref{lem1}, the values of $\eta_i^{(4,Q)}$ for $i=0,1,2,3$ are known. When $p\equiv3({\rm mod~} 4)$, we obtain
\begin{eqnarray*}
  w_b(\ccc(\beta)) &=& \frac{(q^b-1)(Q-1)}{q^bN}-\frac{4\eta_0^{(4,Q)}}{q^bN}\#U(b,j,4) \\
  ~ &~& -\frac{4\eta_{j_1}^{(4,Q)}}{q^bN}(q^b-1-\#U(b,j,4))~~~({\hbox {where $j_1
  \neq 0$}}) \\
  ~ &=& \left\{
        \begin{array}{ll}
          \frac{(q^b-1)(Q-Q^{\frac{1}{2}})}{q^bN}
      +\frac{4Q^{\frac{1}{2}}}{q^bN}\#U(b,j,4), & \hbox{if $\frac{sm}{2}$ is even;} \\
      \frac{(q^b-1)(Q+Q^{\frac{1}{2}})}{q^bN}
      -\frac{4Q^{\frac{1}{2}}}{q^bN}\#U(b,j,4), & \hbox{if $\frac{sm}{2}$ is odd.}
        \end{array}
      \right.
\end{eqnarray*}
This completes the proof.
\end{proof}
The values of Gaussian periods of some small order, such as, $\eta_{i}^{(5,Q)}$, $\eta_{i}^{(6,Q)}$, $\eta_{i}^{(8,Q)}$ and $\eta_{i}^{(12,Q)}$ are determined in \cite{Gura, Hoshi}. Mimicking the proofs above, we can give the $b$-symbol weight enumerators of $\C(Q,N)$ when $N_1=5,6,8$ and $12$. These results depend on the values of $\#U(b,j,N_1)$ and have very tedious expression of the Gaussian periods. For the sake of brevity, we omit these cases.

\subsection{The complete $b$-symbol weight enumerator in the semi-primitive cases}
The following theorem gives the $b$-symbol weight enumerators for a class of irreducible cyclic codes.
\begin{theorem}\label{thm16}
Let $s\cdot m$ be even and $N_1>2$. Assume that there exists a positive integer $j$ such that $p^j\equiv-1({{\rm{mod}}~N_1})$, and the $j$ is the least such. Let $Q=p^{2j\gamma}$ for some integer $\gamma$.

If $\gamma, p$ and $\frac{p^j+1}{N_1}$ are all odd, we have the following result. Assume that $N_1\leq  Q^{\frac{1}{2}}$.
\begin{itemize}
  \item If $\beta=0$, then $w_b(\ccc(\beta))=0;$
  \item If $\beta\in C_i^{(N_1,Q)}$, then
  $$w_b(\ccc(\beta))=\frac{(q^b-1)(Q+Q^{\frac{1}{2}})}{q^bN}-\frac{N_1Q^{\frac{1}{2}}}{q^bN}
  \#U(b,j_1,N_1),$$
  where $i+j_1\equiv\frac{N_1}{2}({\rm mod~} N_1).$
\end{itemize}
Moreover, the $b$-symbol weight enumerator of $\C(Q,N)$ is
$$
A_b(T)=1+\frac{Q-1}{N_1}\sum_{j=1}^{N_1}T^{u_j},
$$
where
$$
u_i=\frac{(q^b-1)(Q+Q^{\frac{1}{2}})}{q^bN}-\frac{N_1Q^{\frac{1}{2}}}{q^bN}
  \#U(b,j_1,N_1).
$$

In all other cases, we have the following. Assume that Assume that $N_1\leq Q^{\frac{1}{2}}$ if $\gamma$ is odd.
\begin{itemize}
  \item If $\beta=0$, then $w_b(\ccc(\beta))=0;$
  \item If $\beta\in C_i^{(N_1,Q)}$, then
  $$w_b(\ccc(\beta))=\frac{(q^b-1)(Q-(-1)^{\gamma}Q^{\frac{1}{2}})}
  {q^bN}+\frac{(-1)^{\gamma}N_1Q^{\frac{1}{2}}}{q^bN}
  \#U(b,j_2,N_1),$$
  where $i+j_2\equiv0({\rm mod~} N_1).$
\end{itemize}
Moreover, the $b$-symbol weight enumerator of $\C(Q,N)$ is
$$
A_b(T)=1+\frac{Q-1}{N_1}\sum_{j=1}^{N_1}T^{u_j},
$$
where
$$
u_i=\frac{(q^b-1)(Q-(-1)^{\gamma}Q^{\frac{1}{2}})}
  {q^bN}+\frac{(-1)^{\gamma}N_1Q^{\frac{1}{2}}}{q^bN}
  \#U(b,j_2,N_1).
$$
\end{theorem}
\begin{proof}
To ensure the dimension of $\C(Q,N)$ is $2j\gamma$, there is not any $\beta\in\F_Q^*$ such that $w_b(\ccc(\beta))=0$. Then $w_b(\ccc(\beta))\geq b$ for any $\beta\in\F_Q^*$.

If $\gamma, p$ and $\frac{p^j+1}{k}$ are all odd, by (\ref{eta7}) and (\ref{eta8}) of Lemma \ref{lem1}, then
\begin{eqnarray*}
  w_b(\ccc(\beta)) &=& \frac{(q^b-1)(Q-1)}{q^bN}-\frac{N_1\eta_{\frac{N_1}{2}}^{(N_1,Q)}}{q^bN}\#U(b,j_1,N_1) \\
  ~ &~& -\frac{N_1\eta_{j_1^{\prime}}^{(N_1,Q)}}{q^bN}(q^b-1-\#U(b,j_1,N_1))~~~({\hbox {where $j_1^{\prime}
  \neq \frac{N_1}{2}$}}) \\
  ~ &=& \frac{(q^b-1)(Q+Q^{\frac{1}{2}})}{q^bN}-\frac{N_1Q^{\frac{1}{2}}}{q^bN}
  \#U(b,j_1,N_1),
\end{eqnarray*}
where $i+j_1\equiv\frac{N_1}{2}({\rm mod~} N_1).$ We need the assumption that $N_1\leq Q^{\frac{1}{2}}= \left\lfloor\frac{Q-b}{Q^{\frac{1}{2}}}\right\rfloor+1$ since $\#U(b,j_1,N_1)\leq q^b-1$ and $w_b(\ccc(\beta))\geq b$.

In all other cases, we obtain
\begin{eqnarray*}
  w_b(\ccc(\beta)) &=& \frac{(q^b-1)(Q-1)}{q^bN}-\frac{N_1\eta_{0}^{(N_1,Q)}}{q^bN}\#U(b,j_2,N_1) \\
  ~ &~& -\frac{N_1\eta_{j_2^{\prime}}^{(N_1,Q)}}{q^bN}(q^b-1-\#U(b,j_2,N_1))
  ~~~({\hbox {where $j_2^{\prime}
  \neq 0$}}) \\
  ~ &=& \frac{(q^b-1)(Q-(-1)^{\gamma}Q^{\frac{1}{2}})}
  {q^bN}+\frac{(-1)^{\gamma}N_1Q^{\frac{1}{2}}}{q^bN}
  \#U(b,j_2,N_1),
\end{eqnarray*}
where $i+j_2\equiv0({\rm mod~} N_1).$ Similarly, we need the assumption that $N_1\leq Q^{\frac{1}{2}}= \left\lfloor\frac{Q-b}{Q^{\frac{1}{2}}}\right\rfloor+1$ to ensure $w_b(\ccc(\beta))\geq b$ if $\gamma$ is odd.
\end{proof}
It is easy to check that Theorem \ref{thm1515}, Theorem \ref{thm11}, and Theorem \ref{thm1717} are special cases of Theorem \ref{thm16}.
When $N_1=N$, this is the classical semi-primitive case.
The Hamming weight enumerator of the semi-primitive irreducible cyclic codes are studied by Delsarte and Goethals \cite{DG}, McEliece \cite{McEliece}, and Baumert and McEliece \cite{BM}. Ding and Yang \cite{DY2013} considered the more flexible case where $N_1=\gcd\left(\frac{Q-1}{q-1},N\right)>2.$ Theorem \ref{thm16} generalizes their results to $b$-symbol metric.

\section{The $b$-symbol weight hierarchy of some irreducible cyclic codes}
Yang {\it et al.} \cite{YLFL} studied Hamming weight hierarchy $\dd_b(\C)$ of irreducible cyclic codes in 2015.
In this section, we study the $b$-symbol weight hierarchy $d_b(\C)$ of some classes of irreducible cyclic codes and compare the two hierarchies. Since $d_b(\C)=\dd_b(\C)$ for any cyclic code if $b=1$ or $b=\dim(\C)$, we omit the two trivial cases in this section.
\subsection{The $b$-symbol weight hierarchy when $N_1=1$}
The following result on the $b$-th generalized weight hierarchy of $\C(Q,N)$ is from \cite[Corollary 7]{YLFL}.
\begin{corollary}\cite{YLFL}\label{YLFLth10}
If $N_1=1$ and $2\leq b\leq m-1$, then
$$\dd_b(\C(Q,N))=\frac{\left(q^b-1\right)Q}{q^bN}.$$
\end{corollary}
Combining Theorem \ref{thm8} and Corollary \ref{YLFLth10}, we obtain the following result directly.
\begin{corollary}
If $N_1=1$ and $2\leq b\leq m-1$, then
$$d_b(\C(Q,N))=\dd_b(\C(Q,N))=
\frac{\left(q^b-1\right)Q}{q^bN}.$$
\end{corollary}
\subsection{The $b$-symbol weight hierarchy when $N_1=2$}
The following result is also from \cite{YLFL}.
\begin{theorem}\label{YLFLthm10}
If $N_1=2$, then $2\leq b\leq m-1$, then
$$\dd_b(\C(Q,N))=\left\{
              \begin{array}{ll}
                \frac{\left(q^b-1\right)\left(Q-Q^{\frac{1}{2}}\right)}{Nq^b}, & \hbox{for $2\leq b\leq \frac{m}{2}$;} \\
                \frac{Qq^b-2Q+q^b}{Nq^b}, & \hbox{for $\frac{m}{2}< b\leq m-1$.}
              \end{array}
            \right.
$$
\end{theorem}
When $N_1=2$, the two types of the minimum distance of the same irreducible cyclic code are equal under some restrictions.
\begin{theorem}
If $N_1=2$ and $1\leq b\leq m-1$, then
$$d_b(\C(Q,N))=\left\{
              \begin{array}{ll}
                \frac{\left(q^b-1\right)\left(Q+Q^{\frac{1}{2}}\right)}{Nq^b}-\frac{2Q^{\frac{1}{2}}}{Nq^b}\#U(b,0,2), & \hbox{if $\#U(b,0,2)\geq \frac{q^b-1}{2}$;} \\
                \frac{\left(q^b-1\right)\left(Q-Q^{\frac{1}{2}}\right)}{Nq^b}+\frac{2Q^{\frac{1}{2}}}{Nq^b}\#U(b,0,2), & \hbox{if $\#U(b,0,2)< \frac{q^b-1}{2}$.}
              \end{array}
            \right.
$$
\end{theorem}
Moreover, $d_b(\C(Q,N))=\dd_b(\C(Q,N))$ if and only if one of the following two statement holds:
\begin{itemize}
  \item [{\rm (1)}]$1\leq b\leq \frac{m}{2}$ and $\#U(b,0,2)=q^b-1.$
  \item [{\rm (2)}]$\frac{m}{2}< b\leq m$ and $\#U(b,i,2)=\frac{(q^b-Q^{\frac{1}{2}})(Q^{\frac{1}{2}}+1)}{2Q^{\frac{1}{2}}}$ for some $i\in \{0,1\}.$
\end{itemize}
\begin{proof}
Assume that $\beta_0\in C_0^{(2,Q)}$ and $\beta_1\in C_1^{(2,Q)}$.
If $\#U(b,0,2)\geq \frac{q^b-1}{2}$, then we have
$$\left\{
  \begin{array}{ll}
    w_b(\ccc(\beta_0)) = \frac{(q^b-1)(Q-Q^{\frac{1}{2}})}{Nq^b}
    +\frac{2Q^{\frac{1}{2}}}{q^bN}\#U(b,0,2) & \hbox{~} \\
    \geq w_b(\ccc(\beta_1))= \frac{(q^b-1)(Q+Q^{\frac{1}{2}})}{Nq^b}-\frac{2Q^{\frac{1}{2}}}{q^bN}\#U(b,0,2), & \hbox{if $p\equiv1({\rm mod~ 4})$ or $\frac{sm}{2}$ is even;} \\
    w_b(\ccc(\beta_0)) = \frac{(q^b-1)(Q+Q^{\frac{1}{2}})}{Nq^b}
    -\frac{2Q^{\frac{1}{2}}}{q^bN}\#U(b,0,2) & \hbox{~} \\
    \geq w_b(\ccc(\beta_1))= \frac{(q^b-1)(Q-Q^{\frac{1}{2}})}{Nq^b}+\frac{2Q^{\frac{1}{2}}}{q^bN}\#U(b,0,2), & \hbox{if $p\equiv3({\rm mod ~4})$ and $\frac{sm}{2}$ is odd.}
  \end{array}
\right.
$$
If $\#U(b,0,2)< \frac{q^b-1}{2}$, then we have
$$\left\{
  \begin{array}{ll}
    w_b(\ccc(\beta_0)) = \frac{(q^b-1)(Q-Q^{\frac{1}{2}})}{NQ^b}
    +\frac{2Q^{\frac{1}{2}}}{q^bN}\#U(b,0,2) & \hbox{~} \\
    <w_b(\ccc(\beta_1))= \frac{(q^b-1)(Q+Q^{\frac{1}{2}})}{NQ^b}
    -\frac{2Q^{\frac{1}{2}}}{q^bN}\#U(b,0,2), & \hbox{if $p\equiv1({\rm mod ~4})$ or $\frac{sm}{2}$ is even;} \\
    w_b(\ccc(\beta_0)) = \frac{(q^b-1)(Q+Q^{\frac{1}{2}})}{NQ^b}
    -\frac{2Q^{\frac{1}{2}}}{q^bN}\#U(b,0,2) & \hbox{~} \\
    >w_b(\ccc(\beta_1))= \frac{(q^b-1)(Q-Q^{\frac{1}{2}})}{NQ^b}+\frac{2Q^{\frac{1}{2}}}{q^bN}\#U(b,0,2), & \hbox{if $p\equiv3({\rm mod ~4})$ and $\frac{sm}{2}$ is odd.}
  \end{array}
\right.
$$
Therefore, the minimum $b$-symbol distance of $\C(Q,N)$ is
$$d_b(\C(Q,N))=\left\{
              \begin{array}{ll}
                \frac{(q^b-1)(Q+Q^{\frac{1}{2}})}{Nq^b}-\frac{2Q^{\frac{1}{2}}}{Nq^b}\#U(b,0,2), & \hbox{if $\#U(b,0,2)\geq \frac{q^b-1}{2}$;} \\
                \frac{(q^b-1)(Q-Q^{\frac{1}{2}})}{Nq^b}+\frac{2Q^{\frac{1}{2}}}{Nq^b}\#U(b,0,2), & \hbox{if $\#U(b,0,2)< \frac{q^b-1}{2}$.}
              \end{array}
            \right.
$$
The two types of the minimum distance are equal if and only if the following equations hold:
\begin{itemize}
  \item If $2\leq b\leq \frac{m}{2}$, then
  $$
  (q^b-1)(Q-Q^{\frac{1}{2}})=\left\{
                               \begin{array}{ll}
                                 (q^b-1)(Q+Q^{\frac{1}{2}})-2Q^{\frac{1}{2}}\#U(b,0,2), & \hbox{if $\#U(b,0,2)\geq \frac{q^b-1}{2};$} \\
                                 (q^b-1)(Q-Q^{\frac{1}{2}})+2Q^{\frac{1}{2}}\#U(b,0,2), & \hbox{if $\#U(b,0,2)< \frac{q^b-1}{2};$}
                               \end{array}
                             \right.
  $$
  $$\Longleftrightarrow \#U(b,0,2)=0~~\hbox{or}~~\#U(b,0,2)=q^b-1.$$
  According to the part (5) of Lemma \ref{lem7}, $\#U(b,0,2)$ can not be zero, then $\#U(b,0,2)=q^b-1.$
  \item If $\frac{m}{2}< b\leq m-1$, then
  $$
  Qq^b-2Q+q^b=\left\{
                               \begin{array}{ll}
                                 (q^b-1)(Q+Q^{\frac{1}{2}})-2Q^{\frac{1}{2}}\#U(b,0,2), & \hbox{if $\#U(b,0,2)\geq \frac{q^b-1}{2};$} \\
                                 (q^b-1)(Q-Q^{\frac{1}{2}})+2Q^{\frac{1}{2}}\#U(b,0,2), & \hbox{if $\#U(b,0,2)< \frac{q^b-1}{2};$}
                               \end{array}
                             \right.
  $$
  $$\Longleftrightarrow \#U(b,0,2)=\frac{(q^b+Q^{\frac{1}{2}})(Q^{\frac{1}{2}}-1)}{2Q^{\frac{1}{2}}}~~
  \hbox{or}~~\#U(b,0,2)=\frac{(q^b-Q^{\frac{1}{2}})(Q^{\frac{1}{2}}+1)}{2Q^{\frac{1}{2}}}.$$
\end{itemize}
This completes the proof.
\end{proof}
\begin{example}
When $q=3$, $m=10$ and $b=2$, we have $\#U(2,0,2)=8=q^b-1.$ Then the minimal symbol-pair distance of $\C(3^{10},2)$ equals $\dd_2(\C(3^{10},2))=\frac{(3^2-1)(3^5-1)\cdot3^5}{2\cdot 3^2}=26136$.
\end{example}
\subsection{The $b$-symbol weight hierarchy in the semi-primitive cases}
\begin{theorem}\label{thm25}
Let $s\cdot m$ be even and $N_1>2$. Assume that there exists a positive integer $j$ such that $p^j\equiv-1({{\rm{mod}}~N_1})$, and the $j$ is the least such. Let $Q=p^{2j\gamma}$ for some integer $\gamma$. Then
\begin{small}
\begin{equation*}
  d_b(\C(Q,N))=\left\{
                 \begin{array}{ll}
                   \frac{(q^b-1)(Q+Q^{\frac{1}{2}})}
  {q^bN} -\frac{N_1Q^{\frac{1}{2}}}{q^bN}\cdot
  \max\{\#U(b,i,N_1)| i \in \mathbb{Z}_{N_1}\}, & \hbox{if $\gamma$ is odd and $N_1\leq Q^{\frac{1}{2}}$,}\\
                   \frac{(q^b-1)(Q-Q^{\frac{1}{2}})}
  {q^bN} +\frac{N_1Q^{\frac{1}{2}}}{q^bN}\cdot
  \min\{\#U(b,i,N_1)| i \in \mathbb{Z}_{N_1}\}, & \hbox{if $\gamma$ is even.}
                 \end{array}
               \right.
\end{equation*}
\end{small}
\end{theorem}
\begin{proof}
From Theorem \ref{thm16}, if $\gamma, p$ and $\frac{p^j+1}{N_1}$ are all odd, we have the following result. Assume that $N_1\leq Q^{\frac{1}{2}}$.
Then
\begin{eqnarray*}
  d_b(\C(Q,N)) &=& \frac{(q^b-1)(Q+Q^{\frac{1}{2}})}{q^bN} -\frac{N_1Q^{\frac{1}{2}}}{q^bN}
  \max\{\#U(b,i,N_1)| i \in \mathbb{Z}_{N_1}\}.
\end{eqnarray*}
In all other cases, we have the following result. Assume that $N_1\leq Q^{\frac{1}{2}}$ if $\gamma$ is odd.
Then
\begin{small}
\begin{equation*}
  d_b(\C(Q,N))=\left\{
                 \begin{array}{ll}
                   \frac{(q^b-1)(Q+Q^{\frac{1}{2}})}
  {q^bN} -\frac{N_1Q^{\frac{1}{2}}}{q^bN}\cdot
  \max\{\#U(b,i,N_1)| i \in \mathbb{Z}_{N_1}\}, & \hbox{if $\gamma$ is odd,}\\
                   \frac{(q^b-1)(Q-Q^{\frac{1}{2}})}
  {q^bN} +\frac{N_1Q^{\frac{1}{2}}}{q^bN}\cdot
  \min\{\#U(b,i,N_1)| i \in \mathbb{Z}_{N_1}\}, & \hbox{if $\gamma$ is even.}
                 \end{array}
               \right.
\end{equation*}
\end{small}
This completes the proof.
\end{proof}
\begin{remark}
From the Magma experimental data, $\#U(b,0,N_1)$ is always the maximum value of the set $\{\#U(b,i,N_1)| i \in \mathbb{Z}_{N_1}\}$. So it is important to determine the value of $\#U(b,0,N_1)$.
\end{remark}
\section{Shortened codes from irreducible cyclic codes}
In this section, we introduce a technique for constructing new codes from old codes, which we call $b$-symbol shortened construction. To this end, we need the following definition of $b$-symbol support of a vector.
\begin{definition}
The $b$-symbol support of a vector $\x$ is defined by
$$\mathcal{I}_b(\x)=supp(\pi_b(\x))=\bigcup_{i=0}^{b-1}supp(\tau^i(\x)),$$
where $supp(\x)$ denotes the support of the vector $\x.$ Let $\overline{\mathcal{I}_b(\x)}
=\{1,2,\ldots,n\}\setminus\mathcal{I}_b(\x).$
\end{definition}
\begin{remark}
In the $b$-th generalized Hamming metric, there is a definition about the support of the subcode $\D$ of $\C$ which defined to be
 \begin{equation*}
   \chi(\D)=\{i:0\leq i\leq n-1|c_i\neq0~ \hbox{for some} ~ (c_0,c_1,\ldots,c_{n-1})\in \D\}.
 \end{equation*}
The two definitions on the support can be viewed as two different generalizations for the support of a vector.
\end{remark}
The following proposition gives an interesting shortening technique.
\begin{prop}\label{prop1}
Let $\C$ be a cyclic code with parameters $[n,K,d_H(\C)]$ over $\F_q$. Let $1\leq b\leq K$ and let $\ccc$ be a codeword with the minimal $b$-symbol weight of $\C$. Then the shortened code $C_{\overline{\mathcal{I}_b(\ccc)}}$ has parameters $[d_b(\C), b, \geq d_H(\C)].$
\end{prop}
\begin{proof}
Let $G_b(\ccc)$ be the matrix
\begin{equation*}
  G_b(\ccc)=\left(
    \begin{array}{c}
      \ccc \\
      \tau(\ccc) \\
      \vdots \\
      \tau^{b-1}(\ccc) \\
   \end{array}
  \right)_{b\times n}
  =\left(
             \begin{array}{cccc}
              \mathbf{r}_1 & \mathbf{r}_2 & \cdots & \mathbf{r}_n \\
            \end{array}
          \right)_{b\times n},
\end{equation*}
where $\mathbf{r}_i$ are column vectors belonging to $ \F_q^b$ for $i \in \{1,\ldots,n\}.$
From \cite[Lemma 16]{BUG}, the rank of $G_b(\ccc)$ equals $b$.
Let
 $$G_b^{\prime}(\ccc)=\left(
                                    \begin{array}{cccc}
                                      \mathbf{r}_{j_1} & \mathbf{r}_{j_2} & \ldots & \mathbf{r}_{j_m} \\
                                    \end{array}
                                  \right)_{b\times m},
$$ where $\mathbf{r}_{j_1},\ldots,\mathbf{r}_{j_m}$ are all the nonzero columns of $G_b(\ccc)$.
By the definition of the $b$-symbol weight metric, $m=d_b(\C).$
 This yields the desired result, since the shortened code $C_{\overline{\mathcal{I}_b(\ccc)}}$ is generated by $G_b^{\prime}(\ccc)$.
\end{proof}
The following bound is the famous Griesmer bound. It was proved by Griesmer \cite{GRIE} for binary codes, and later generalized by Solomon and Stiffler \cite{SS} for $q>2$.
\begin{theorem}\label{thm31}
Let $\C$ be a linear codes with parameters $[n,K,d_H(\C)]$ over $\F_q$ with $K\geq 1$. Then $$
n\geq \sum_{i=0}^{K-1}\left\lceil\frac{d_H(\C)}{q^i}
\right\rceil.
$$
\end{theorem}
A code which achieves the Griesmer bound is called a Griesmer code.
 The following theorem gives a class of Griesmer codes from $\C(Q,N)$ that shorten some proper coordinates.
\begin{theorem}\label{short1}
Let $N_1=1$ and $1\leq b\leq m$. Let $\ccc$ be a nonzero codeword of $\C(Q,N)$. Then the parameters of the shortened codes $\C(Q,N)_{\overline{\mathcal{I}_b(\ccc)}}$ are  $$\left[\frac{(q^b-1)Q}{q^bN},b,\frac{(q-1)Q}{qN}\right].$$
Moreover, $\C(Q,N)_{\overline{\mathcal{I}_b(\ccc)}}$ are Griesmer codes if $N|q-1$.
\end{theorem}
\begin{proof}
When $N_1=1$, any nonzero codeword of $\C(Q,N)$ has the minimal $b$-symbol weight. Combining Proposition \ref{prop1} and Theorem \ref{thm8}, we obtain the desired parameters. Recall that $Q=q^m$. Then we have
\begin{eqnarray*}
  \sum_{i=0}^{b-1}\left\lceil\frac{\frac{(q-1)Q}{qN}}{q^i}\right\rceil=\frac{q-1}{N}\sum_{i=1}^{b}q^{m-i}=n
  .
\end{eqnarray*}
Therefore, $\C(Q,N)_{\overline{\mathcal{I}_b(\ccc)}}$ are Griesmer codes if $N|q-1$.
\end{proof}
\begin{remark}
When $N=q-1$, the code $\C(Q,N)$ is the Simplex code.
Liu {\it et al.} \cite{LDT} considered the code from the Simplex code $S(m,q)$ which shortened any one or two coordinates. Theorem \ref{short1} consider the case where the size of shortening set is greater than $2$.
\end{remark}
\begin{example}
As we all know, the parameters of the Simplex code $S(m,q)$ over
 $\F_q$
  are
   $[\frac{q^m-1}{q-1},m,q^{m-1}].$
Let $\ccc$ be a nonzero codeword of $S(m,q)$.
The parameters of the shortened codes $S(m,q)_{\overline{\mathcal{I}_b(\ccc)}}$ are given in  Table 2. These codes are all Griesmer codes. From the numerical example, we can see that when the size of the shortened set exceeds $2$, we still get the codes with very good parameters.
\begin{table}[H]
  \centering
  \caption{Numeral examples of Theorem \ref{thm31}}\label{A}

\begin{tabular}{c|c}
  \hline
   Shortened code & Parameters  \\
  \hline
  $S(4,2)_{\overline{\mathcal{I}_3(\ccc)}}$ & $[14,3,8]_2$ \\
  $S(4,2)_{\overline{\mathcal{I}_2(\ccc)}}$ & $[12,2,8]_2$ \\
  $S(5,2)_{\overline{\mathcal{I}_4(\ccc)}}$ & $[30,4,16]_2$ \\
  $S(5,2)_{\overline{\mathcal{I}_3(\ccc)}}$ & $[28,3,16]_2$ \\
  $S(5,2)_{\overline{\mathcal{I}_2(\ccc)}}$ & $[24,2,16]_2$ \\
  $S(4,3)_{\overline{\mathcal{I}_3(\ccc)}}$ & $[39,3,27]_3$ \\
  $S(4,3)_{\overline{\mathcal{I}_2(\ccc)}}$ & $[36,2,27]_3$ \\
  $S(5,3)_{\overline{\mathcal{I}_4(\ccc)}}$ & $[120,4,81]_3$ \\
  $S(5,3)_{\overline{\mathcal{I}_3(\ccc)}}$ & $[117,3,81]_3$ \\
  $S(5,3)_{\overline{\mathcal{I}_2(\ccc)}}$ & $[108,2,81]_3$ \\
  $S(4,4)_{\overline{\mathcal{I}_3(\ccc)}}$ & $[84,3,64]_4$ \\
  $S(4,4)_{\overline{\mathcal{I}_2(\ccc)}}$ & $[80,2,64]_4$ \\
  $S(5,4)_{\overline{\mathcal{I}_4(\ccc)}}$ & $[340,4,256]_4$ \\
  $S(5,4)_{\overline{\mathcal{I}_3(\ccc)}}$ & $[336,3,256]_4$ \\
  $S(5,4)_{\overline{\mathcal{I}_2(\ccc)}}$ & $[320,2,256]_4$ \\
  \hline
\end{tabular}
\end{table}
\end{example}
\section{Summary and concluding remarks}
The main contributions of this paper are the following:
\begin{itemize}
  \item A general formula for computing the $b$-symbol weight of a nonzero codeword of an irreducible cyclic code is given. It is a generalization of the formula for computing the Hamming weight of a nonzero codeword of an irreducible cyclic code from \cite{DY2013}.
  \item The $b$-symbol weight hierarchies of some irreducible cyclic codes are given. The two types of weight hierarchies of the same irreducible cyclic code mentioned in this paper are compared. In particular, the two weight hierarchies are equal under certain conditions.
  \item  we present an application of the $b$-symbol weight hierarchy
of cyclic codes in the shortening technique and construct some new shortened codes with
nice parameters. Compared to the work in \cite{LDT}, when the object code is an irreducible cyclic code, the size of the shortened set can be greater than $2$.
 \item The results that the $b$-symbol weight hierarchies of irreducible cyclic codes provide nice upper bounds on the generalized weight hierarchies of irreducible cyclic codes.
\end{itemize}
\section*{Acknowledgement}
This research is supported by Natural Science Foundation of China (12071001), Excellent Youth Foundation of Natural Science Foundation of Anhui Province (1808085J20).


\begin{thebibliography}{1}
\bibitem{BM} L. D. Baumert, R. J. McEliece, Weight of irreducible cyclic codes,  Inform. Contr., {\bf 20}(2), (1972), 158--175.
\bibitem{Ber}B. C. Berndt, R. J. Evans, K. S. Williams, Gauss and Jacobi Sums, J. Wiley and Sons Company, New York, 1997.
\bibitem{CB1} Y. Cassuto, M. Blaum, Codes for symbol-pair read channels, In: Proc. IEEE Int. Symp. Inf. Theory, Austin, TX, USA, (2010), 988--992.
\bibitem{CB} Y. Cassuto, M. Blaum, Codes for symbol-pair read channels,  IEEE Trans. Inf. Theory, {\bf 57}(12), (2011), 8011--8020.
\bibitem{CL} Y. Cassuto, S. Litsyn, Symbol-pair codes: algebraic constructions and asymptotic bounds, In: Proc. IEEE Int. Symp. Inf. Theory, St. Petersburg, Russia, (2011), 2348--2352.
\bibitem{C+} Y. M. Chee, L. Ji, H. M. Kiah, C. Wang, J. Yin, Maximum distance separable codes for symbol pair read channels, IEEE Trans. Inf. Theory, {\bf 59}(11), (2013), 7259--7267.
\bibitem{C+1} Y. M. Chee, H. M. Kiah, C. Wang, J. Yin, Maximum distance separable symbol-pair codes, In: Proc. IEEE Int. Symp. Inf. Theory, Cambridge, MA, USA, (2012), 2886--2890.
\bibitem{DG} P. Delsarte, J. M. Goethals, Irreducible binary cyclic codes of even dimension, in: Proc. Second Chapel Hill Conf. on Combinatorial Mathematics and its Applications, Univ. North Carolinam Chapel Hill, NC, (1970), 100--113.
\bibitem{CLL} B. Chen, L. Lin, H. Liu, Constacyclic symbol-pair codes: lower bounds and optimal constructions, IEEE Trans. Inf. Theory, {\bf 63}(12), (2017), 7661--7666.
\bibitem{D} P. Delsarte, Four fundamentals parameters of a code and their combinatorial significance, Inform. Contr., {\bf 23}, (1973), 407--438.
\bibitem{DGZ}B. Ding, G. Ge, J. Zhang, T. Zhang, Y. Zhang, New constructions of MDS symbol-pair codes, Des. Codes Cryptogr., {\bf 86}(4), (2018), 841--859.
\bibitem{DZG} B. Ding, T. Zhang, G. Ge, Maximum distance separable codes for $b$-symbol read channels, Finite Fields Appl., {\bf 49}, (2018), 180--197.

\bibitem{DY2013} C. Ding, J. Yang, Hamming weights in irreducible cyclic codes,  Discr. Math., {\bf 313}(4), (2013), 434--446.
\bibitem{Eli} O. Elishco, R. Gabrys, E. Yaakobi, Bounds and constructions of codes over symbol-pair read channels, IEEE Trans. Inf. Theory, {\bf 66}(3), (2020), 1385--1395.
\bibitem{GRIE} J. H. Griesmer, A bound for error-correcting codes, IBM J. Res. Dev. {\bf 4}, (1960), 532--542.
\bibitem{Gura} S. J. Gurak, Periodic polynomials for $\F_q$ of fixed small degree, CRM Proc. Lecture Notes, {\bf {36}}, (2004), 127--145.
\bibitem{Hell2} T. Helleseth, T. Kl${\o}$ve, J. Mykkeltveit, The weight distribution of irreducible cyclic codes with block lengths $n_1((q^l-1)/N)$, Discr. Math., {\bf 18}(2), (1977), 179--211.
\bibitem{Hoshi} A. Hoshi, Explicit lifts of quintic Jacobi sums and periodic polynomials for $\F_q$. Prco. Japan Acad. Ser. A, {\bf 82}, (2006), 87--92.
\bibitem{KZL} X. Kai, S. Zhu, P. Li, A construction of new MDS symbol-pair codes, IEEE Trans. Inf. Theory, {\bf 61}(11), (2015), 5828--5834.
\bibitem{K1} T. Kl{\o}ve, The weight distribution of linear codes over ${\rm GF}(q^l)$ having generator matrix over $GF(q)$, Discrete Math., {\bf 23}(2), (1978), 159--168.
\bibitem{LG}S. Li, G. Ge, Constructions of maximum distance separable symbol-pair codes using cyclic and constacyclic codes. Des. Codes Cryptogr. {\bf84}, (2017), 359--372.
\bibitem{LP1}H. Liu, X. Pan, Generalized pair weights of linear codes and linear isomorphism preserving pair weights, IEEE Trans. Inf. Theory, doi: 10.1109/TIT.2021.3120229.
\bibitem{LDT}Y. Liu, C. Ding, Shorten linear codes over finite fields, IEEE Trans. Inf. Theory, {\bf67}(8), (2021), 5119--5132.
\bibitem{ML1}J. Ma, J. Luo, On symbol-pair weight distribution of MDS codes and simplex codes over finite fields. Cryptogr. Commun., {\bf13}, (2021), 101--115.
\bibitem{ML}J. Ma, J. Luo, MDS symbol-pair codes from repeated-root cyclic codes. Des. Codes Cryptogr., (2021). doi:10.1007/s10623-021-00967-4.
\bibitem{McEliece}R. J. McEliece, A class of two-weight codes, Jet Propulsion Laboratory Space Program Summary, 37--41, vol. IV, 264--166.
\bibitem{LN}R. Lidl, H. Niederreiter,  Finite fields, volume {\bf 20}, Cambridge University Press, (1997).
\bibitem{McE}R. J. McEliece, Irreducible cyclic codes and Gauss sums, combinatorics, in: Proc. NATO Advanced study Inst., Breuklen, 1974, Part 1: Theory of Designs, Finite Geometry and Coding Theory, in: Math. Centre Tracts, vol.55, Math. Centrum, Amsterdam, 1974, 179--196.
\bibitem{Mye} G. Myerson, Period polynomials and Gauss sums for finite fields, Acta Arith. {\bf 39}, (1981), 251--264.
\bibitem{LP2}X. Pan, Generalized $b$-weights and $b$-MDS codes, arXiv:2103.16299.



\bibitem{SOS} M. Shi, F. \"{O}zbudak, P. Sol\'{e}, Geometric approach to $b$-symbol Hamming weights of cyclic codes,
    IEEE Trans. Inf. Theory, {\bf 67}(6), (2021), 3735--3751.
\bibitem{BUG}M. Shi, H. Zhu, T. Helleseth, The connections among Hamming metric, $b$-symbol metric, and $r$-th generalized Hamming metric, arXiv: 2109.13764.
\bibitem{SS} G. Solomon, J. J. Stiffler, Algebraically punctured cyclic codes, Inform. Contr., {\bf 8}, (1965), 170--179.
\bibitem{Stor}T. Stor, Cyclotomy and Different Sets, Markham, Chicago, 1967.
\bibitem{SZW}Z. Sun, S. Zhu, L. Wang, The symbol-pair distance distribution of a class
of repeated-root cyclic codes over $\F_{p^m},$ Cryptogr. Commun., {\bf 10}(4), (2018), 643--653.

\bibitem{wei} V. K. Wei, Generalized Hamming weights for linear codes, IEEE Trans. Inf. Theory, {\bf 37}(5), (1991), 1412--1418.
\bibitem{Yaa} E. Yaakobi, J. Bruck, P. H. Siegel, Decoding of cyclic codes over symbol-pair read channels, In: Proc. IEEE Int. Symp. Inf. Theory, Cambridge, MA, USA, (2012), 2891--2895.
\bibitem{Yaa1} E. Yaakobi, J. Bruck, P. H. Siegel, Constructions and decoding of cyclic codes over $b$-symbol real channels,  IEEE Trans. Inf. Theory, {\bf 62}(4), (2016), 1541--1551.
\bibitem{YLFL} M. Yang, J. Li, K. Feng, D. Lin, Generalized Hamming weights of irreducible cyclic codes,
IEEE Trans. Inf. Theory, {\bf 61}(9), (2015), 4905--4913.
\bibitem{YLF} M. Yang, J. Li, K. Feng, Construction of cyclic and constacyclic codes for $b$-symbol read channels meeting the Plotlkin-like bound, arXiv: 1607. 02677.
\bibitem{ZHW} H. Zhu, M. Shi, F. \"{O}zbudak, Complete $b$-symbol weight distribution of some irreducible cyclic codes, Des. Codes Cryptogr. (2022). https://doi.org/10.1007/s10623-022-01030-6.
\bibitem{ZHW1} H. Zhu, M. Shi,  The $b$-symbol weight hierarchy of the Kasami codes, arXiv: 2112.04019.
\bibitem{ZHW2} H. Zhu, M. Shi,  How many distinct $b$-symbol distances can a code have? submitted.
\end{thebibliography}
\end{document}